\documentclass[11pt,a4paper]{article}
\usepackage[margin=25mm]{geometry}
\usepackage[T1]{fontenc}
\usepackage{lmodern}
\usepackage[utf8]{inputenc}
\usepackage{amsmath,amssymb,amsthm,mathtools,bm}
\usepackage{booktabs,tabularx,array,enumitem}
\usepackage{microtype,tikz}
\usepackage[colorlinks=true,linkcolor=blue!40!black,citecolor=blue!40!black,urlcolor=blue!50!black]{hyperref}
\setlist{itemsep=2pt,topsep=3pt,leftmargin=*}
\newtheorem{theorem}{Theorem}
\newtheorem{lemma}[theorem]{Lemma}
\newtheorem{proposition}[theorem]{Proposition}
\newtheorem{corollary}[theorem]{Corollary}
\theoremstyle{remark}\newtheorem{remark}[theorem]{Remark}
\newcommand{\Tr}{\operatorname{Tr}}
\newcommand{\id}{\operatorname{id}}
\newcommand{\B}{\mathcal B}
\newcommand{\D}{\mathcal D}
\newcommand{\N}{\mathbb N}
\newcommand{\R}{\mathbb R}
\newcommand{\C}{\mathbb C}
\newcommand{\ket}[1]{|#1\rangle}
\newcommand{\bra}[1]{\langle#1|}
\newcommand{\vac}{\ket{0}\!\bra{0}}
\newcommand{\norm}[1]{\left\lVert#1\right\rVert}
\newcommand{\dd}{\mathrm d}
\hypersetup{pdftitle={Equality in the Bosonic Quantum Entropy Power Inequality},
pdfsubject={Finite-energy equality characterization for bosonic beam-splitter entropy inequalities}}
\title{Equality in the Bosonic Quantum Entropy Power Inequality}
\author{Yonglong Li\\School of Mathematics and Statistics\\ Xi'an Jiaotong University\\ liyonglong@xjtu.edu.cn}
\date{}
\begin{document}
\maketitle
\begin{abstract}
The bosonic quantum entropy power inequality bounds the entropy of a
beam-splitter output in terms of the input entropies. We determine its
complete equality class among independent inputs of finite mean photon
number: for any number of modes and any transmissivity strictly between
zero and one, equality in either the linear or the exponential form
holds if and only if the inputs are Gaussian with the same covariance
matrix, allowing arbitrary displacements. The main step shows that
equality for one output forces the two outputs to be independent:
a thermal auxiliary converts equality into preservation of mutual
information by a fixed noisy channel, and a two-outcome instrument
shows that preservation of mutual information requires a product state
with the reference.
The quantum Darmois--Skitovich theorem then gives Gaussianity, while
a central-limit argument reduces exponential equality to linear equality.
\end{abstract}

\section{Introduction}\label{sec:introduction}

Consider two independent $m$-mode bosonic systems, $A$ and $B$, mixed
on a beam splitter of transmissivity $0<\eta<1$. If $R_X$ denotes the
vector of position and momentum quadratures of system $X$, the outputs
are
\begin{align}
 R_C&=\sqrt\eta R_A+\sqrt{1-\eta}R_B,\notag\\
 R_E&=-\sqrt{1-\eta}R_A+\sqrt\eta R_B.
 \label{eq:intromixing}
\end{align}
The quantum entropy power inequality (qEPI) lower-bounds the entropy
of the retained output $C$:
\begin{align}
 S(C)&\geq\eta S(A)+(1-\eta)S(B),\label{eq:introlinear}\\
 e^{S(C)/m}&\geq\eta e^{S(A)/m}+(1-\eta)e^{S(B)/m}.
 \label{eq:introexponential}
\end{align}
Here $S$ is the von Neumann entropy, with natural logarithms. We ask:
which independent inputs make either bound exact?

K\"onig and Smith~\cite{KS,KScorrection} proved
\eqref{eq:introlinear} for every transmissivity and
\eqref{eq:introexponential} for a balanced beam splitter.
De Palma, Mari, and Giovannetti~\cite{DPMG} proved the exponential
inequality for every transmissivity. The extension to several independent
inputs under general linear mixing is due to De Palma, Mari, Lloyd,
and Giovannetti~\cite{DPMultimode}. These results provide the
inequalities but not a complete finite-energy equality
characterization. De Palma \emph{et al.} observed that Gaussian inputs
with proportional covariance matrices fail to attain the exponential
bound unless their entropies agree~\cite[discussion following Eq.~(5)]{DPMG}.
We establish the complete characterization here.

Equality questions distinguish this quantum problem from its classical
counterpart. For independent random vectors with finite covariance
matrices and finite differential entropies, equality in the classical
EPI of Shannon and Stam~\cite{Stam,DCT} characterizes Gaussian inputs
with proportional covariances.
The linear-entropy form requires identical covariances
\cite[Sec.~II-E1]{Rioul}. In the quantum problem, both forms require
identical covariance matrices.

\subsection{Main results}

Finite energy means finite expectation of total photon number.
It guarantees finite covariance matrices and finite input and output
entropies. Let $U_\eta$ denote the beam-splitter unitary in
\eqref{eq:intromixing}.

\begin{theorem}[Linear equality]\label{thm:main}
Let $\rho_A,\rho_B$ be independent finite-energy $m$-mode states,
where $m\geq1$ and $0<\eta<1$, and set
$\omega_{CE}=U_\eta(\rho_A\otimes\rho_B)U_\eta^\dagger$.
The following statements are equivalent:
\begin{enumerate}[label=(\roman*)]
\item $S(C)=\eta S(A)+(1-\eta)S(B)$;
\item $\omega_{CE}=\omega_C\otimes\omega_E$;
\item $\rho_A,\rho_B$ are Gaussian states with the same covariance matrix.
\end{enumerate}
The input displacement vectors may differ.
\end{theorem}

The equivalence between product outputs and Gaussian inputs is the
quantum Darmois--Skitovich characterization
\cite[Theorem~7, Corollary~8]{Cuesta}, with a one-mode predecessor
in~\cite{Springer}. The new implication is (i)$\Rightarrow$(ii).
It applies to arbitrary mixed inputs and includes Gaussian states
with symplectic eigenvalues equal to $1/2$.

\begin{theorem}[Exponential equality]\label{thm:exp}
Let $\rho_A,\rho_B$ be independent finite-energy $m$-mode states,
where $m\geq1$ and $0<\eta<1$. For the retained beam-splitter output $C$,
\[
 e^{S(C)/m}=\eta e^{S(A)/m}+(1-\eta)e^{S(B)/m}
\]
if and only if the inputs are Gaussian with the same covariance matrix,
allowing arbitrary displacements.
\end{theorem}

Thus the two qEPIs have the same equality class. In particular, when
both input covariance matrices are fixed to the same matrix, the
centered Gaussian pair uniquely minimizes the linear deficit
(Corollary~\ref{cor:fixedcov}). For unequal input entropies the qEPI
bound is strict at every finite-energy input pair.

This equality question is distinct from the entropy photon-number
inequality (EPnI)~\cite{GES}, which asks whether product thermal inputs minimize
output entropy for prescribed input entropies. The EPnI is a stronger
conjectured bound; see~\cite[Secs.~V.2 and~VII.1]{GaussianReview}.
Identical thermal inputs attain the qEPI, whereas thermal inputs
with unequal entropies give a strict inequality.

\subsection{A first calculation: the two output deficits}

The following identity makes the missing implication precise. Define
\begin{align}
 L_C&=S(C)-\eta S(A)-(1-\eta)S(B),\notag\\
 L_E&=S(E)-(1-\eta)S(A)-\eta S(B).
 \label{eq:twoDeficits}
\end{align}
Both quantities are nonnegative. For $E$, apply the linear qEPI after
the parity transformation $R_A\mapsto-R_A$, which preserves entropy.
Since the input is a product and the beam splitter is unitary,
\begin{align}
 I(C;E)_\omega
 &=S(C)+S(E)-S(CE)\notag\\
 &=S(C)+S(E)-S(A)-S(B)\notag\\
 &=L_C+L_E. \label{eq:outputdeficit}
\end{align}
Consequently, product outputs force both deficits to vanish.
Theorem~\ref{thm:main} proves the converse: $L_C=0$ already forces
$L_E=0$. Exchanging the roles of the two outputs (transmissivity
$1-\eta$ together with the parity map on $A$), the same theorem shows
that $L_E=0$ forces $L_C=0$.

For pure inputs this implication is immediate. The equality $L_C=0$
becomes $S(C)=0$, so the $C$ marginal is pure and the joint output
is a product. The difficulty is to obtain the same conclusion when
both input states may be mixed.

\subsection{Related work and the proof idea}

For identical copies at $\eta=1/2$, Bu and
Li~\cite[Theorems~5 and~7]{BuLi} characterize Gaussianity by product
outputs. In the finite-energy class, their criterion is equivalently
$S(C)+S(E)=2S(\rho)$. They also give a pure-state
criterion~\cite[Theorem~1]{BuLi}.
Hahn and Takagi~\cite[Theorem~1]{HT} construct a faithful
non-Gaussianity measure from the correlations generated by
identical-copy mixing. Our starting condition involves a single output,
and permits distinct mixed inputs and any nontrivial transmissivity.
For discrete-variable systems, Bu, Gu, and Jaffe~\cite{BGJ}
characterized the minimum-output-entropy inputs of a qudit convolution
as pure stabilizer states; the present paper concerns the bosonic
equality problem.

The proof uses a doubling construction, an approach familiar from
classical Gaussian extremality~\cite{GengNair}. Classical arguments
combine the entropy chain rule with conditioning on random variables.
In the quantum circuit, the corresponding entropy decomposition gives
a vanishing conditional mutual information. This alone does not imply
unconditional independence. Furthermore, a quantum memory cannot
generally be replaced by a family of classical conditional input
distributions without changing the state.

Our choice of auxiliary inputs resolves this particular obstruction.
We adjoin two identical thermal states and evaluate the doubled circuit
in two orders; see Fig.~\ref{fig:circuit}. One order gives the entropy
decomposition. In the other, the thermal pair supplies independent noise
on the two original outputs. The zero conditional mutual information
then states that one local noisy channel loses no mutual information.

The channel can prepare a fixed vacuum output with an
input-dependent probability. We expose that event with a two-outcome
quantum instrument. In the vacuum branch, every correlation with a
reference is discarded. If the channel loses no mutual information,
the filtered input in that branch must already be a product.
The filter is injective, so the original input must be a product as
well. An injectivity argument on the other local channel then yields
independence of the original beam-splitter outputs.

The inequality used in the doubling step is the conditional qEPI.
K\"onig~\cite{KoenigConditional} proved a conditional linear inequality
at weight $p=\eta$ for two independent Gaussian input--memory pairs,
and conjectured a general conditional inequality.
De Palma and Trevisan~\cite{DPT} established the conditional qEPI
under conditional independence. Their Gaussian construction
asymptotically attains the bound for prescribed limiting conditional
entropies~\cite[Theorem~7]{DPT}. Our unconditional equality
classification is obtained from the noisy-channel argument.

For exponential equality, we keep one affine weight fixed under repeated
convolution. The doubling identity preserves zero deficit at each
iteration. The quantum central limit theorem~\cite{CH,BDLR} and
entropy continuity~\cite{Winter} then yield Gaussian equality.
Strictness within the Gaussian family reduces the question to
Theorem~\ref{thm:main}.

In Section~\ref{sec:definitions} we state our conventions and the
conditional qEPI. In Section~\ref{sec:linearproof} we develop the two tools
and prove linear equality. In Section~\ref{sec:exp} we prove exponential
equality. In Section~\ref{sec:discussion} we discuss further questions.
The appendices contain the auxiliary operator calculations and a
self-contained finite-energy Darmois--Skitovich proof.

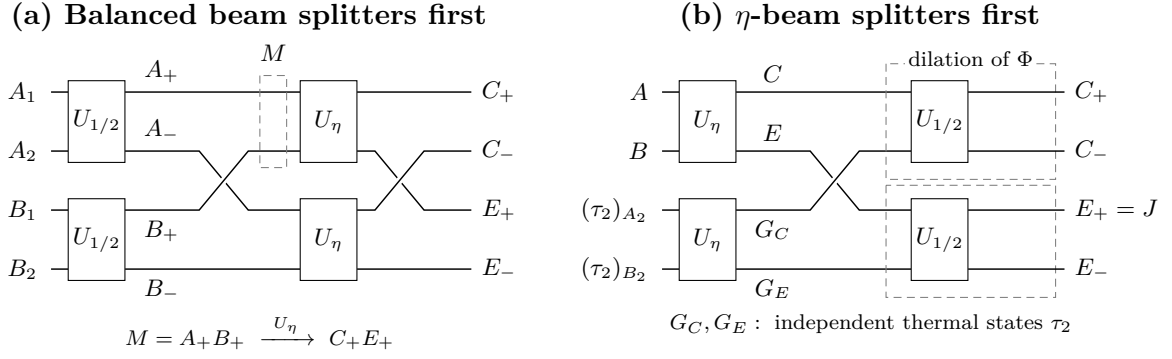
\begin{figure}[tbp]
\centering
\begin{minipage}[t]{.495\textwidth}
\centering
\textbf{(a) Balanced beam splitters first}\par\smallskip
\begin{tikzpicture}[x=1.18cm,y=.78cm,
 every node/.style={font=\footnotesize},
 gate/.style={draw,fill=white,minimum width=.75cm,
  minimum height=1.08cm,inner sep=1pt},
 wire/.style={line width=.55pt}]
 \path[use as bounding box] (-.95,-4) rectangle (5.5,.85);
 \node[anchor=east] at (0,0) {$A_1$};
 \node[anchor=east] at (0,-1) {$A_2$};
 \node[anchor=east] at (0,-2) {$B_1$};
 \node[anchor=east] at (0,-3) {$B_2$};
 \node[gate] (a) at (.55,-.5) {$U_{1/2}$};
 \node[gate] (b) at (.55,-2.5) {$U_{1/2}$};
 \draw[wire] (.04,0)--(.23,0) (.04,-1)--(.23,-1)
  (.04,-2)--(.23,-2) (.04,-3)--(.23,-3);
 \draw[wire] (.87,0)--(2.83,0);
 \draw[wire] (.87,-3)--(2.83,-3);
 \draw[wire] (.87,-1)--(1.70,-1)--(2.25,-2)--(2.83,-2);
 \draw[white,line width=2.6pt] (1.70,-2)--(2.25,-1);
 \draw[wire] (.87,-2)--(1.70,-2)--(2.25,-1)--(2.83,-1);
 \node[anchor=south] at (1.28,0) {$A_+$};
 \node[anchor=south] at (1.28,-1) {$A_-$};
 \node[anchor=north] at (1.28,-2) {$B_+$};
 \node[anchor=north] at (1.28,-3) {$B_-$};
 \draw[densely dashed,gray] (2.38,.28) rectangle (2.68,-1.28);
 \node[anchor=south] at (2.53,.35) {$M$};
 \node[gate] (p) at (3.15,-.5) {$U_\eta$};
 \node[gate] (q) at (3.15,-2.5) {$U_\eta$};
 \draw[wire] (3.47,0)--(4.75,0) node[right] {$C_+$};
 \draw[wire] (3.47,-3)--(4.75,-3) node[right] {$E_-$};
 \draw[wire] (3.47,-1)--(3.67,-1)--(4.22,-2)--(4.75,-2)
  node[right] {$E_+$};
 \draw[white,line width=2.6pt] (3.67,-2)--(4.22,-1);
 \draw[wire] (3.47,-2)--(3.67,-2)--(4.22,-1)--(4.75,-1)
  node[right] {$C_-$};
 \node[anchor=north,font=\scriptsize] at (2.38,-3.6)
  {$M=A_+B_+\ \xrightarrow{\ U_\eta\ }\ C_+E_+$};
\end{tikzpicture}
\end{minipage}\hfill
\begin{minipage}[t]{.495\textwidth}
\centering
\textbf{(b) $\eta$-beam splitters first}\par\smallskip
\begin{tikzpicture}[x=1.18cm,y=.78cm,
 every node/.style={font=\footnotesize},
 gate/.style={draw,fill=white,minimum width=.75cm,
  minimum height=1.08cm,inner sep=1pt},
 wire/.style={line width=.55pt}]
 \path[use as bounding box] (-.95,-4) rectangle (5.5,.85);
 \node[anchor=east] at (0,0) {$A$};
 \node[anchor=east] at (0,-1) {$B$};
 \node[anchor=east] at (0,-2) {$(\tau_2)_{A_2}$};
 \node[anchor=east] at (0,-3) {$(\tau_2)_{B_2}$};
 \node[gate] (a) at (.55,-.5) {$U_\eta$};
 \node[gate] (b) at (.55,-2.5) {$U_\eta$};
 \draw[wire] (.04,0)--(.23,0) (.04,-1)--(.23,-1)
  (.04,-2)--(.23,-2) (.04,-3)--(.23,-3);
 \draw[wire] (.87,0)--(2.83,0);
 \draw[wire] (.87,-3)--(2.83,-3);
 \draw[wire] (.87,-1)--(1.70,-1)--(2.25,-2)--(2.83,-2);
 \draw[white,line width=2.6pt] (1.70,-2)--(2.25,-1);
 \draw[wire] (.87,-2)--(1.70,-2)--(2.25,-1)--(2.83,-1);
 \node[anchor=south] at (1.28,0) {$C$};
 \node[anchor=south] at (1.28,-1) {$E$};
 \node[anchor=north] at (1.28,-2) {$G_C$};
 \node[anchor=north] at (1.28,-3) {$G_E$};
 \draw[densely dashed,gray] (2.55,.48) rectangle (4.45,-1.48);
 \draw[densely dashed,gray] (2.55,-1.58) rectangle (4.45,-3.48);
 \node[gate] (p) at (3.15,-.5) {$U_{1/2}$};
 \node[gate] (q) at (3.15,-2.5) {$U_{1/2}$};
 \draw[wire] (3.47,0)--(4.55,0) node[right] {$C_+$};
 \draw[wire] (3.47,-1)--(4.55,-1) node[right] {$C_-$};
 \draw[wire] (3.47,-2)--(4.55,-2) node[right] {$E_+=J$};
 \draw[wire] (3.47,-3)--(4.55,-3) node[right] {$E_-$};
 \node[anchor=south,font=\scriptsize,fill=white,inner sep=1pt]
  at (3.5,.39) {dilation of $\Phi$};
 \node[anchor=north,font=\scriptsize] at (2.38,-3.6)
  {$G_C,G_E:\ \text{independent thermal states }\tau_2$};
\end{tikzpicture}
\end{minipage}
\par\vspace{.8ex}
\caption{Two orders of the doubled circuit. Each wire is an $m$-mode
system; crossings only route systems. In (a), the upper
$\eta$-beam splitter gives the plus-pair deficit $F_p(A_+,B_+)$, the
lower one gives $F_p(A_-,B_-\mid M)$, and the remaining term in
Lemma~\ref{lem:doubling} is $I(C_-;E_+\mid C_+)$.
In (b), the auxiliary inputs are identical thermal states $\tau_2$ with
mean photon number $2$ per mode. Their $\eta$-beam splitter leaves
them independent and unchanged. Each dashed enclosure is a dilation of
$\Phi(\rho)=\B_{1/2}(\rho\otimes\tau_2)$: retaining its upper output and
discarding its lower output gives $C\mapsto C_+$ or $E\mapsto E_+$.
Both outputs of each dilation remain displayed because the proof also
uses their joint state. The two circuit orders produce the same joint
state on $C_+C_-E_+E_-$.}
\label{fig:circuit}
\end{figure}

\section{Conventions and the conditional qEPI}\label{sec:definitions}

\subsection{Bosonic systems and entropy}

We use the standard bosonic framework~\cite{HW,Weedbrook}. An $m$-mode system has Hilbert space
$\mathcal H_m=\ell^2(\N_0^m)$, number basis $\ket{\bm n}$, annihilation operators $a_j$, and quadratures
\[
 Q_j=\frac{a_j+a_j^\dagger}{\sqrt2},\qquad
 P_j=\frac{a_j-a_j^\dagger}{i\sqrt2}.
\]
Write
\[
 R=(Q_1,P_1,\ldots,Q_m,P_m)^T,\qquad
 \Omega=\bigoplus_{j=1}^m\begin{pmatrix}0&1\\-1&0\end{pmatrix}.
\]
Then $[R_j,R_k]=i\Omega_{jk}I$, where $\Omega$ is the symplectic
form, with $\hbar=1$ and vacuum covariance $I/2$.

A state is a positive trace-class operator of trace one. We write $\rho_X=\Tr_Y\rho_{XY}$ for a marginal and $\norm T_1=\Tr\sqrt{T^\dagger T}$ for trace norm. A channel is a completely positive trace-preserving map on trace-class operators; $\id$ denotes the identity channel. The total photon-number operator is
\begin{equation}\label{eq:Hidentity}
 H_m=\sum_{j=1}^m a_j^\dagger a_j
 =\frac12 R^TR-\frac m2 I.
\end{equation}
For a finite-energy state $\rho$, its mean vector $d_\rho$ and covariance matrix $V_\rho$
are
\begin{align}
 d_\rho&=\Tr\rho R,\notag\\
 (V_\rho)_{jk}
 &=\frac12\Tr\rho\{R_j-(d_\rho)_jI,R_k-(d_\rho)_kI\},
 \label{eq:moments}
\end{align}
where $(d_\rho)_k$ is the $k$-th element of the vector $d_\rho$ and $\{X,Y\}=XY+YX$ is the anticommutator of operators $X$ and $Y$. Unbounded products are
interpreted as quadratic-form expectations; details are in
Appendix~\ref{sec:analytic}. A state is centered if $d_\rho=0$, and
\begin{equation}\label{eq:energycov}
 \Tr\rho H_m=\frac{\operatorname{tr}V_\rho+|d_\rho|^2-m}{2}.
\end{equation}
By~\eqref{eq:energycov}, finite energy is the same as finite second
moments of the quadratures, the hypothesis used in the quantum central
limit theorem~\cite{BDLR,BGM}.

All logarithms are natural. The von Neumann entropy of a state $\rho$ is $S(\rho)=-\Tr\rho\log\rho$,
with the convention that $0\log0=0$. Define
\[
 g(x)=(x+1)\log(x+1)-x\log x,\qquad x\geq0.
\]
The oscillator Gibbs bound~\cite{Winter,HW} is
\begin{equation}\label{eq:energyentropy}
 \Tr\rho H_m\leq\mathsf E
 \quad\Longrightarrow\quad
 S(\rho)\leq m g(\mathsf E/m)<\infty.
\end{equation}
Each marginal of a finite-mode, finite-energy state has finite entropy.
All beam splitters below preserve total photon number, so this bound
applies throughout the circuits.

For finite entropies, our conventions for mutual information, conditional
entropy, and conditional mutual information (CMI) are
\begin{align}
 I(X;Y)&=S(X)+S(Y)-S(XY),\notag\\
 S(X\mid M)&=S(XM)-S(M),\label{eq:conditionalS}\\
 I(X;Y\mid M)&=S(X\mid M)-S(X\mid YM).
 \label{eq:CMIdef}
\end{align}
In general $I(X;Y)=D(\rho_{XY}\Vert\rho_X\otimes\rho_Y)$, where
$D(\rho\Vert\sigma)=\Tr\rho(\log\rho-\log\sigma)$ is the quantum
relative entropy~\cite{Lindblad}, which is nonnegative and vanishes if
and only if $\rho=\sigma$; in particular $I(X;Y)=0$ exactly for product
states.
We use data processing~\cite{Lindblad,MR}, strong
subadditivity~\cite{LR}, and the chain rule
$I(XY;K)=I(X;K)+I(Y;K\mid X)$.

\subsection{Characteristic functions and Gaussian states}

The real Weyl characteristic function is
\[
 W(\xi)=e^{i\xi^TR},\qquad
 \chi_\rho(\xi)=\Tr\rho W(\xi),\qquad \xi\in\R^{2m}.
\]
It determines the state. Gaussian states are precisely those for which
\begin{equation}\label{eq:GaussianDef}
 \chi_\rho(\xi)=\exp(i d^T\xi-\tfrac12\xi^TV\xi).
\end{equation}
A real symmetric matrix $V$ is the covariance matrix of some quantum
state if and only if $V+i\Omega/2\geq0$~\cite{Weedbrook}. Write $\gamma_V$ for the unique centered Gaussian
state with covariance $V$. Its Williamson form and entropy are
\begin{align}
 V&=\mathsf S\operatorname{diag}(\nu_1,\nu_1,\ldots,\nu_m,\nu_m)
       \mathsf S^T,\notag\\[-2pt]
 &\hspace{15mm}\mathsf S\Omega\mathsf S^T=\Omega,\quad \nu_j\geq\tfrac12,
 \label{eq:Williamson}\\
 S(\gamma_V)&=\sum_{j=1}^m g(\nu_j-\tfrac12).
 \label{eq:GaussianEntropy}
\end{align}
The matrix $\mathsf S$ is symplectic and $\nu_j$ are the symplectic
eigenvalues~\cite[Sec.~III.C, Eq.~(3.27)]{HW}.

For heterodyne calculations we also use the complex convention
\begin{align}
 \D(z)&=\exp\!\left(\sum_j z_ja_j^\dagger-\bar z_ja_j\right),
 &\widehat\rho(z)&=\Tr\rho\D(z),\label{eq:complexD}\\
 W(\xi)&=\D(z(\xi)),
 &z(\xi)_j&=\frac{-\xi_{2j}+i\xi_{2j-1}}{\sqrt2}.
 \label{eq:realcomplex}
\end{align}
Here $z\in\C^m$, and $\ket z=\D(z)\ket0$ is a coherent state.
Displacement changes the mean and preserves covariance and entropy.
After matching quadrature ordering, De Palma--Trevisan's real argument
is related to ours by $\xi=\Omega x$; Cuesta's is related by
$\xi=-\Omega x$, and his covariance is $\Gamma=2V$.\footnote{Theorem
and equation numbers in citations refer to the arXiv versions of the
cited works; their identifiers are given in the bibliography.}

The retained beam-splitter channel is
\begin{equation}\label{eq:Bchannel}
 \B_\eta(T_{AB})=\Tr_E(U_\eta T_{AB}U_\eta^\dagger).
\end{equation}
For a product input, $\B_\eta(\rho_A\otimes\rho_B)$ is the quantum
convolution $\rho_A\boxplus_\eta\rho_B$ of~\cite{BDLR,BM}, so
condition~(i) of Theorem~\ref{thm:main} reads
$S(\rho_A\boxplus_\eta\rho_B)=\eta S(\rho_A)+(1-\eta)S(\rho_B)$.
For independent inputs, its joint-output characteristic function is
\begin{align}
 \chi_{CE}(u,v)
 ={}&\chi_A(\sqrt\eta\,u-\sqrt{1-\eta}\,v)\chi_B(\sqrt{1-\eta}\,u+\sqrt\eta\,v).
 \label{eq:BSchi}
\end{align}
The same mixing formula holds for complex arguments.
In particular $V_C=\eta V_A+(1-\eta)V_B$.

\subsection{The affine deficit}

The proof uses a variational weight $p$, distinct from the physical
transmissivity $\eta$. For $0\leq p\leq1$ and $0<q<1$, let
\begin{equation}\label{eq:binaryD}
 d(p\Vert q)=p\log\frac pq+(1-p)\log\frac{1-p}{1-q}.
\end{equation}
For independent inputs define
\begin{align}
 F_p(A,B)&=S(C)-pS(A)-(1-p)S(B)+m d(p\Vert\eta),\label{eq:F}\\
 L_\eta(A,B)&=F_\eta(A,B)=L_C. \label{eq:linear}
\end{align}
For a joint state $\rho_{ABM}$, with
$\rho_{CM}=(\B_\eta\otimes\id_M)(\rho_{ABM})$, define
\begin{align}
 F_p(A,B\mid M)
 ={}&S(C\mid M)-pS(A\mid M)-(1-p)S(B\mid M)+m d(p\Vert\eta).
 \label{eq:Fc}
\end{align}
Write $L_\eta(A,B\mid M)=F_\eta(A,B\mid M)$ for the conditional
linear deficit. The conditional qEPI~\cite[Theorem~6, Eq.~(106)]{DPT} states that
\begin{equation}\label{eq:importedconditional}
 I(A;B\mid M)=0\quad\Longrightarrow\quad F_p(A,B\mid M)\geq0
\end{equation}
for every $p\in[0,1]$, provided the bosonic inputs have finite energy
and $S(M)<\infty$. A one-dimensional memory gives the unconditional
inequality. Below the memory consists of finite-energy bosonic modes.

\begin{table}[t]
\caption{Notation used in the proof}
\label{tab:notation}
\centering
\begin{tabularx}{\columnwidth}{@{}lX@{}}
\toprule
Symbol & Meaning\\
\midrule
$\eta,\ p$ & Physical transmissivity and variational weight\\
$R,\Omega;\ d,V$ & Quadratures, symplectic form; mean and covariance\\
$H_m,\ \tau_N$ & Photon number and thermal state with $N$ photons per mode\\
$L_C,L_E;\ F_p$ & Two output deficits; affine deficit\\
$M=A_+B_+$ & Memory in the doubling identity\\
$\Xi$ & Final four-output state of the doubled circuit\\
$\Phi$ & Thermal attenuator of transmissivity $1/2$ with environment $\tau_2$\\
$G,\ F$ & Vacuum-branch effect and binary flag\\
$J=E_+,\ \sigma_{CJ}$ & Retained noisy reference and noisy joint state\\
$\rho^{\boxplus n}$ & $n$-fold quantum convolution of $\rho$\\
\bottomrule
\end{tabularx}
\end{table}

\section{Proof of linear equality}\label{sec:linearproof}

We first establish an entropy decomposition and a strict information-loss
property. Their combination will prove Theorem~\ref{thm:main}.

\subsection{Conditional doubling}\label{sec:doubling}

Start with four independent finite-energy $m$-mode input systems in the state
\begin{equation}
 \rho_{A_1}\otimes\rho_{B_1}\otimes\rho_{A_2}\otimes\rho_{B_2}.
\end{equation}
The two pairs may be different. Mix $A_1,A_2$ on a balanced beam splitter,
and independently mix $B_1,B_2$ in the same way:
\begin{equation}\label{eq:copy}
 R_{A_+}=\frac{R_{A_1}+R_{A_2}}{\sqrt2},\qquad
 R_{A_-}=\frac{-R_{A_1}+R_{A_2}}{\sqrt2},
\end{equation}
with analogous definitions for $B$. We call these two balanced
beam splitters the copy step. Then apply the $\eta$-beam splitter
$U_\eta$ to $A_+B_+$ and to $A_-B_-$, producing $C_+E_+$ and $C_-E_-$.
Set $M=A_+B_+$ before converting it unitarily to $C_+E_+$.
To specify the states underlying these labels, write $\zeta$ for the state
after the copy step:
\begin{align}\label{eq:zeta}
 \zeta_{A_+A_-B_+B_-}
 ={}&\bigl[U_{1/2}(\rho_{A_1}\otimes\rho_{A_2})U_{1/2}^\dagger\bigr]
\otimes\bigl[U_{1/2}(\rho_{B_1}\otimes\rho_{B_2})U_{1/2}^\dagger\bigr].
\end{align}
Tensor factors are reordered according to their labels. Applying
$\B_\eta$ to the minus pair and retaining $M$ gives a state on $C_-M$.
Applying the $\eta$-beam splitter to $M$ then gives the same state on
$C_-C_+E_+$.
The latter is a marginal of the four-output state, denoted
$\Xi_{C_+E_+C_-E_-}$. The labels $M$ and $C_+E_+$ describe the
memory before and after its $\eta$-beam splitter, respectively.

This construction is a quantum version of the doubling trick used in
classical Gaussian extremality arguments~\cite{GengNair}. Here the
conditional qEPI provides the nonnegative remainder. Mixing two input pairs separates
their total deficit into the deficit of the plus pair, the deficit of
the minus pair conditioned on the plus inputs, and a correlation term.
The following identity makes all three terms nonnegative. Allowing the two input pairs to differ will let us
choose a thermal auxiliary in Section~\ref{sec:mainproof}.

\begin{lemma}[Doubling with different pairs]\label{lem:doubling}
Let $\zeta_{A_-B_-M}$ be the state in~\eqref{eq:zeta}, with tensor factors
reordered and $M=A_+B_+$. Define
\[
 \zeta_{C_-M}=(\B_\eta\otimes\id_M)(\zeta_{A_-B_-M}).
\]
The conditional deficit below uses these two states; the CMI is evaluated
in the four-output state $\Xi$. For each fixed $p\in[0,1]$,
\begin{align}\label{eq:general-double}
 &F_p(A_1,B_1)+F_p(A_2,B_2)=F_p(A_+,B_+)+F_p(A_-,B_-\mid M)+I(C_-;E_+\mid C_+).
\end{align}
All three terms on the right are nonnegative.
\end{lemma}

\begin{proof}
The product in \eqref{eq:zeta} proves independence between the entire
$A$-group and the entire $B$-group. Indeed,
\begin{align}
 S(A_-\mid A_+B_+)=S(A_-A_+)+S(B_+)-S(A_+)-S(B_+)=S(A_-\mid A_+).
\end{align}
Similarly $S(B_-\mid A_+B_+)=S(B_-\mid B_+)$. Moreover
\begin{equation}
 S(A_-B_-\mid A_+B_+)
 =S(A_-\mid A_+)+S(B_-\mid B_+),
\end{equation}
which verifies conditional independence explicitly. Hence
\begin{align}\label{eq:admissible}
 I(A_-;B_-\mid A_+B_+)&=0,\notag\\
 S(A_-\mid M)&=S(A_-\mid A_+),\notag\\
 S(B_-\mid M)&=S(B_-\mid B_+).
\end{align}
The $A$ entropy terms in the right side of \eqref{eq:general-double} satisfy
\begin{align}
 S(A_+)+S(A_-\mid M)&=S(A_+A_-)\notag\\
 &=S(A_1A_2)=S(A_1)+S(A_2),
\end{align}
where the final equality uses independence of the original inputs.
Multiplication by $-p$ gives the desired $A$ contribution; the $B$ argument
is identical with coefficient $-(1-p)$.
The output terms combine as
\begin{align}
 &S(C_+)+S(C_-\mid M)+I(C_-;E_+\mid C_+)\notag\\
 &\qquad=S(C_+)+S(C_-\mid C_+)\notag\\
 &\qquad=S(C_+C_-).\label{eq:Cchain}
\end{align}
The equality of the two conditional entropies follows by applying the
$\eta$-beam splitter to the memory in both entropy terms:
\begin{align}
 S(C_-\mid M)_\zeta
 &=S(C_-M)_\zeta-S(M)_\zeta\notag\\
 &=S(C_-C_+E_+)_\Xi-S(C_+E_+)_\Xi\notag\\
 &=S(C_-\mid C_+E_+)_\Xi.
\end{align}

The copy step commutes with the $\eta$-beam splitters: the former
acts on the pair index and the latter on the $A/B$ index. To check this directly, put
$a=\sqrt\eta$, $b=\sqrt{1-\eta}$ and first define
$R_{C_i}=aR_{A_i}+bR_{B_i}$ and
$R_{E_i}=-bR_{A_i}+aR_{B_i}$ for $i=1,2$. Substitution shows
\begin{equation}
 R_{C_+}=\frac{R_{C_1}+R_{C_2}}{\sqrt2},\qquad
 R_{C_-}=\frac{-R_{C_1}+R_{C_2}}{\sqrt2},
\end{equation}
with the same identities for $E_+,E_-$. Both circuit orderings therefore
have the same action on all Weyl operators and the same joint output state.
If the $\eta$-beam splitters act first,
$C_1,C_2$ are independent and the copy step acts unitarily on them.
Therefore $S(C_+C_-)=S(C_1)+S(C_2)$. The two constant terms on the right
are $2m\,d(p\Vert\eta)$, proving the identity.

Since all initial energies are finite and passive beam splitters
preserve total energy, every joint entropy here is finite. The first term is nonnegative by unconditional affine qEPI, the second by
conditional affine qEPI and \eqref{eq:admissible}, and the third by strong
subadditivity. \end{proof}

The consequence needed for linear equality is obtained by setting $p=\eta$:
\begin{equation}\label{eq:zeroJ}
 \begin{aligned}
 L_\eta(A_1,B_1)=L_\eta(A_2,B_2)=0\\
 {}\Longrightarrow I(C_-;E_+\mid C_+)=0.
 \end{aligned}
\end{equation}

\subsection{A thermal channel with a common vacuum component}\label{sec:strict}

For $N>0$, define the $m$-mode product thermal state with mean photon number
$N$ per mode by
\begin{equation}\label{eq:tau}
 \tau_N=\bigotimes_{j=1}^m
 \left[\frac1{N+1}\sum_{n=0}^\infty
 \left(\frac{N}{N+1}\right)^n\ket n\!\bra n\right].
\end{equation}
Set $\tau_0=\vac$ when $N=0$. Its covariance is $(N+\frac12)I_{2m}$.
Fix once and for all the channel
\begin{equation}\label{eq:Phi}
 \Phi(\rho)=\B_{1/2}(\rho\otimes\tau_2).
\end{equation}
This is the thermal attenuator of transmissivity $1/2$ with
environment $\tau_2$. The environment energy $N=2$ makes it a
measure-and-prepare channel: a heterodyne measurement followed by
preparation of a displaced $\tau_{1/2}$ state. Such a
channel is entanglement breaking: when applied to one part of a joint
state, its output is a mixture of product states. We use the more specific
fact that every prepared state contains a component along the vacuum.
Any bath mean photon number $N>1$ gives a similar representation,
with prepared thermal mean $(N-1)/2>0$. The choice $N=2$ keeps the
constants simple. The transmissivity of $\Phi$ is $1/2$, independent
of the physical $\eta$.

\subsubsection{Measure-and-prepare representation}

For $z\in\C^m$, the \emph{coherent state} is
\begin{equation}\label{eq:coherent}
 \ket z=\D(z)\ket0
 =e^{-|z|^2/2}\sum_{\bm n\in\N_0^m}
 \frac{z_1^{n_1}\cdots z_m^{n_m}}{\sqrt{n_1!\cdots n_m!}}\ket{\bm n}.
\end{equation}
It satisfies
$\langle u|v\rangle=\exp[-(|u|^2+|v|^2)/2+\bar u\cdot v]$.
The \emph{heterodyne measurement} is the positive-operator-valued measure
(POVM)
\begin{equation}\label{eq:hetPOVM}
 M_{\rm het}(\dd z)=\pi^{-m}\ket z\bra z\,\dd^{2m}z,\qquad
 \int_{\C^m}M_{\rm het}(\dd z)=I.
\end{equation}
A POVM assigns positive operators to measurable outcome sets and sums to
the identity; $\dd^{2m}z$ denotes Lebesgue measure in the real and imaginary
coordinates of $\C^m$. The second identity in \eqref{eq:hetPOVM} follows
by inserting \eqref{eq:coherent} and integrating its number-basis entries.
For a state $\rho$, the outcome density is its Husimi
$Q$-function~\cite{HolevoBook},
\begin{equation}
 Q_\rho(z)=\pi^{-m}\bra z\rho\ket z,
 \qquad \int_{\C^m}Q_\rho(z)\,\dd^{2m}z=1.
\end{equation}

Holevo~\cite[Theorem~2, Eqs.~(30)--(31)]{HolevoEB} gives the
general measure-and-prepare form of entanglement-breaking Gaussian
channels. The representation below is the special case for $\Phi$; the
direct verification fixes the environment photon number and the
prepared-state noise in our convention.

\begin{lemma}\label{lem:heterodyne}
For every input state,
\begin{align}\label{eq:MP}
 \Phi(\rho)&=\int_{\C^m}Q_\rho(z)\,\pi_z\,\dd^{2m}z,\notag\\
 \pi_z&=\D(z/\sqrt2)\tau_{1/2}\D(z/\sqrt2)^\dagger.
\end{align}
\end{lemma}

\begin{proof}
Use the complex displacement characteristic function
$\widehat\rho(w)=\Tr\rho\D(w)$. The heterodyne Fourier identity is
\begin{equation}\label{eq:Qfourier}
 \int Q_\rho(z)e^{\sum_j(w_j\bar z_j-\bar w_jz_j)}\,\dd^{2m}z
 =\widehat\rho(w)e^{-|w|^2/2}.
\end{equation}
The thermal characteristic function is
\[
 \widehat{\tau_N}(w)=e^{-(N+1/2)|w|^2},
\]
as follows from its Gaussian covariance $(N+1/2)I$ and
\eqref{eq:realcomplex}. For a displaced state,
\[
 \Tr[\D(\alpha)\rho\D(\alpha)^\dagger\D(w)]
 =e^{w\cdot\bar\alpha-\bar w\cdot\alpha}\widehat\rho(w).
\]
Let $\Phi_{\rm prep}$ denote the integral map on the right of~\eqref{eq:MP}.
The three factors in its characteristic function are
\begin{align}
 \widehat{\Phi_{\rm prep}(\rho)}(w)
 &=e^{-|w|^2}\int Q_\rho(z)
 e^{(w\cdot\bar z-\bar w\cdot z)/\sqrt2}\,\dd^{2m}z\notag\\
 &\overset{\mathrm{(a)}}= e^{-|w|^2}\widehat\rho(w/\sqrt2)e^{-|w|^2/4}\notag\\
 &=\widehat\rho(w/\sqrt2)e^{-5|w|^2/4}.
\end{align}
Step (a) uses the Fourier identity~\eqref{eq:Qfourier} at $w/\sqrt2$;
Appendix~\ref{sec:FourierProof} gives its coherent-state verification.
On the other hand, the beam-splitter formula gives
\begin{equation}\label{eq:PhiCharacteristic}
 \widehat{\Phi(\rho)}(w)
 =\widehat\rho(w/\sqrt2)\widehat{\tau_2}(w/\sqrt2)
 =\widehat\rho(w/\sqrt2)e^{-5|w|^2/4}.
\end{equation}
Uniqueness of characteristic functions proves~\eqref{eq:MP}.
The integral is in trace norm; its justification is in
Appendix~\ref{sec:analytic}.
\end{proof}

\subsubsection{The vacuum component}

\begin{lemma}\label{lem:vacuum}
Each state in \eqref{eq:MP} satisfies
\begin{equation}\label{eq:domination}
 \pi_z\geq c(z)\vac,\qquad
 c(z)=\left(\frac23\right)^m e^{-|z|^2}\in(0,1).
\end{equation}
Thus $\pi_z=c(z)\vac+(1-c(z))\kappa_z$ for a measurable family of states
$\kappa_z$.
\end{lemma}

\begin{proof}
The thermal formula gives $\tau_{1/2}=(2/3)^m3^{-H_m}$. With
$\alpha=z/\sqrt2$,
\begin{align}
 \bra0\pi_z^{-1}\ket0
 &=\left(\frac32\right)^m\bra{-\alpha}3^{H_m}\ket{-\alpha}\notag\\
 &=\left(\frac32\right)^m e^{2|\alpha|^2}
 =\left(\frac32\right)^m e^{|z|^2}.\label{eq:inversevac}
\end{align}
For one mode the required series is
\begin{equation}
 \langle\alpha|3^{a^\dagger a}|\alpha\rangle
 =e^{-|\alpha|^2}\sum_{n=0}^\infty\frac{3^n|\alpha|^{2n}}{n!}
 =e^{2|\alpha|^2};
\end{equation}
multiplication over the modes proves the middle equality of
\eqref{eq:inversevac}. Each prepared thermal state has zero kernel. The inverse in that equation is interpreted as its
quadratic form, namely $\norm{\pi_z^{-1/2}\ket0}^2$.
In particular $\pi_z^{-1/2}\ket0$ is a well-defined Hilbert
space vector. For every vector $v$, Cauchy--Schwarz gives
\begin{equation}
 |\bra0 v\rangle|^2
 =|\langle\pi_z^{-1/2}0,\pi_z^{1/2}v\rangle|^2
 \leq\bra0\pi_z^{-1}\ket0\,\langle v,\pi_zv\rangle.
\end{equation}
This is the operator inequality \eqref{eq:domination} with the reciprocal
of \eqref{eq:inversevac}. Define
$\kappa_z=(\pi_z-c(z)\vac)/(1-c(z))$. It is positive and has trace
one. The map $z\mapsto\pi_z$ is trace-norm continuous, as follows from
strong continuity of displacements acting on a trace-class operator.
The function $c(z)$ is continuous and $1-c(z)\geq1-(2/3)^m>0$.
Thus $z\mapsto\kappa_z$ is trace-norm continuous, in particular measurable.
\end{proof}

\subsection{A two-outcome instrument and strict information loss}
\label{sec:instrument}

The vacuum component in Lemma~\ref{lem:vacuum} allows us to split $\Phi$
into two branches. In one branch the channel prepares the vacuum and
therefore discards all correlations with a reference system. We shall
show that this branch has positive probability for every input and that
its preceding filter cannot turn a correlated state into a product.
These two facts give the required equality criterion.

Integrating the vacuum weight against the heterodyne POVM defines the
\emph{effect} (a positive operator bounded by the identity)
\begin{equation}\label{eq:vacuumEffect}
 G=\int_{\C^m}c(z)M_{\rm het}(\dd z)
   =3^{-m}2^{-H_m},\qquad 0<G\leq 3^{-m}I.
\end{equation}
Here $0<G$ means that $G$ has zero kernel, not that it has a positive
uniform lower bound. To check the formula, in one mode insert the
coherent-state expansion and integrate in polar coordinates:
\begin{equation}\label{eq:effectIntegral}
 \frac1\pi\int_{\C}e^{-|z|^2}\ket z\bra z\,\dd^2z
 =\sum_{n=0}^{\infty}2^{-(n+1)}\ket n\bra n.
\end{equation}
Multiplication over the modes and the factor $(2/3)^m$ in $c(z)$ give
\eqref{eq:vacuumEffect}.
Equation~\eqref{eq:effectIntegral} is the Glauber--Sudarshan
$P$-representation of the thermal state with one photon per
mode~\cite{Glauber}; thus $G=(2/3)^m\tau_1$ is a multiple of a thermal
state.

Define two completely positive maps by
\begin{align}\label{eq:channelBranches}
 \Lambda_1(\rho)&=\Tr(G\rho)\vac,\notag\\
 \Lambda_0(\rho)&=\int_{\C^m}Q_\rho(z)
                 [\pi_z-c(z)\vac]\,\dd^{2m}z.
\end{align}
The kernel in the second integral is positive by
Lemma~\ref{lem:vacuum}; hence it defines a completely positive
measure-and-prepare map. Its trace is
$\Tr\Lambda_0(\rho)=\Tr[(I-G)\rho]$, and
$\Phi=\Lambda_1+\Lambda_0$. Complete positivity of the integral maps is verified in
Appendix~\ref{sec:instrumentBounds}.

For $f=0,1$, let
\begin{equation}\label{eq:instrumentFilters}
 M_1=G^{1/2},\qquad M_0=(I-G)^{1/2},\qquad
 \mathcal I_f(\rho)=M_f\rho M_f.
\end{equation}
Since $M_0^2+M_1^2=I$, the map
\begin{equation}\label{eq:binaryInstrument}
 \mathcal M(\rho)=\sum_{f=0}^1
       \ket f\bra f_F\otimes\mathcal I_f(\rho)
\end{equation}
is a channel. It records a two-valued outcome in the classical register
$F$ and retains the filtered input, denoted $X'$. Moreover,
$I-G\geq(1-3^{-m})I$, so $(I-G)^{-1/2}$ is bounded and
\begin{equation}\label{eq:residualChannel}
 \Psi_0(\sigma)=\Lambda_0\bigl((I-G)^{-1/2}
                      \sigma(I-G)^{-1/2}\bigr)
\end{equation}
is a channel: it is completely positive, and its trace is $\Tr\sigma$.
After the instrument, replace $X'$ by the vacuum when $F=1$ and apply
$\Psi_0$ when $F=0$. Discarding $F$ then gives exactly $\Phi$, since
$\Psi_0\circ\mathcal I_0=\Lambda_0$.

\begin{proposition}[Strict information loss for $\Phi$]\label{prop:strictMI}
Let $\omega_{XK}$ be a state such that $X$ consists of $m\geq1$ bosonic
modes with $\Tr(\omega_XH_m)<\infty$, and the reference system $K$ has a
separable Hilbert space with $S(\omega_K)<\infty$. Write
$\zeta_{YK}=(\Phi\otimes\id_K)(\omega_{XK})$. Then
\begin{equation}\label{eq:MIrigidity}
 I(X;K)_\omega=I(Y;K)_\zeta
 \quad\Longleftrightarrow\quad
 \omega_{XK}=\omega_X\otimes\omega_K.
\end{equation}
\end{proposition}

\begin{proof}
Let $p_f=\Tr(M_f\omega_XM_f)$. Equation~\eqref{eq:vacuumEffect}
implies $p_1>0$, because all its number-basis eigenvalues are positive,
and $p_0\geq1-3^{-m}>0$. The conditional states after the instrument are
\begin{equation}\label{eq:filteredStates}
 \omega^f_{X'K}
   =p_f^{-1}(M_f\otimes I_K)\omega_{XK}(M_f\otimes I_K).
\end{equation}
Let $\zeta^0_{YK}=(\Psi_0\otimes\id_K)(\omega^0_{X'K})$ and
$\zeta^1_{YK}=\vac_Y\otimes\omega_K^1$. Both transformations preserve
the corresponding $K$ marginal. The state with the flag retained is
$\zeta_{FYK}=\sum_f p_f\ket f\bra f_F\otimes\zeta^f_{YK}$.

All branch and joint entropies below are finite, as verified in
Appendix~\ref{sec:instrumentBounds}. Thus ordinary entropy chain rules apply.

Data processing for $\mathcal M$, followed by the chain rule for a
classical flag, gives
\begin{align}\label{eq:instrumentDPI}
 I(X;K)_\omega
 &\overset{\mathrm{(a)}}\geq I(FX';K)_{(\mathcal M\otimes\id_K)(\omega)}\notag\\
 &= I(F;K)+p_1I(X';K)_{\omega^1}
               +p_0I(X';K)_{\omega^0}.
\end{align}
On the other hand, discarding the flag and then using data processing
for $\Psi_0$ gives
\begin{align}\label{eq:flagOutputDPI}
 I(Y;K)_\zeta
 &\overset{\mathrm{(b)}}\leq I(FY;K)_{\zeta}\notag\\
 &= I(F;K)+p_0 I(Y;K)_{\zeta^0}\notag\\
 &\overset{\mathrm{(c)}}\leq I(F;K)+p_0 I(X';K)_{\omega^0}.
\end{align}
Steps (a), (b), and (c) apply data processing to the instrument,
discarding the flag, and the channel $\Psi_0$, respectively.
The $FK$ marginal is the same in both displays, and the $f=1$ term
vanishes in the second display because $\zeta^1$ is a product.
Subtracting yields
\begin{equation}\label{eq:strictloss}
 I(X;K)_\omega-I(Y;K)_\zeta
       \geq p_1I(X';K)_{\omega^1}.
\end{equation}
If the left side is zero, then $\omega^1=\omega^1_{X'}\otimes\omega^1_K$.
Taking its $X'$ marginal in~\eqref{eq:filteredStates} therefore gives
\begin{equation}\label{eq:filteredProduct}
 (G^{1/2}\otimes I_K)\omega_{XK}(G^{1/2}\otimes I_K)
   =(G^{1/2}\omega_XG^{1/2})\otimes\omega_K^1.
\end{equation}
For number-basis vectors $\ket{\bm n},\ket{\bm n'}$, their eigenvalues
$g_{\bm n}=3^{-m}2^{-\sum_j n_j}$ and $g_{\bm n'}$ are positive.
Taking the corresponding operator block of~\eqref{eq:filteredProduct}
and cancelling $\sqrt{g_{\bm n}g_{\bm n'}}$ gives
\[
 \bra{\bm n}\omega_{XK}\ket{\bm n'}
 =\bra{\bm n}\omega_X\ket{\bm n'}\,\omega_K^1.
\]
Equality of all blocks implies $\omega_{XK}=\omega_X\otimes\omega_K^1$.
Taking the partial trace over $X$ identifies $\omega_K^1=\omega_K$.
The converse holds because a local channel preserves product states.
\end{proof}

\begin{remark}
The filter $G^{1/2}$ is injective, although its inverse is unbounded.
The number-basis argument above uses only its positive eigenvalues;
it does not apply the inverse to a trace-class operator. This distinction
is useful when seeking quantitative versions of~\eqref{eq:strictloss}.
We also note why the standard equality condition for data
processing does not suffice here. By Petz's theorem~\cite{Petz,HJPW},
equality in~\eqref{eq:MIrigidity} implies that $\omega_{XK}$ is
recovered from $\zeta_{YK}$ by a local channel on $Y$; since $\Phi$ is
entanglement breaking, this only shows that $\omega_{XK}$ is separable,
whereas the instrument argument yields a product.
\end{remark}

\subsection{From equality to independent outputs}\label{sec:mainproof}

\begin{proof}[Proof of (i)$\Rightarrow$(ii) in Theorem~\ref{thm:main}]
Let $\omega_{CE}=U_\eta(\rho_A\otimes\rho_B)U_\eta^\dagger$ and assume
$L_\eta(A,B)=0$.

\paragraph{1. Add a thermal equality pair.}
In Lemma~\ref{lem:doubling}, set
\begin{equation}\label{eq:thermalaux}
 (\rho_{A_1},\rho_{B_1})=(\rho_A,\rho_B),\qquad
 (\rho_{A_2},\rho_{B_2})=(\tau_2,\tau_2).
\end{equation}
All four inputs are independent. If the original energies are $\mathsf E_A,\mathsf E_B$,
the total energy in this circuit is $\mathsf E_A+\mathsf E_B+4m$. Every beam splitter preserves
that total, so the energy bound~\eqref{eq:energyentropy} makes all circuit entropies finite.
The auxiliary state has the form
\[
 (\tau_2)_{A_2}\otimes(\tau_2)_{B_2}
 =3^{-2m}\exp[-\log(3/2)(H_{A_2}+H_{B_2})].
\]
Since $U_\eta$ preserves $H_{A_2}+H_{B_2}$, it preserves this state:
\begin{equation}\label{eq:expandedthermalout}
 U_\eta(\tau_2\otimes\tau_2)U_\eta^\dagger
   =(\tau_2)_{G_C}\otimes(\tau_2)_{G_E}.
\end{equation}
In particular, the auxiliary outputs are independent and
$L_\eta(A_2,B_2)=m g(2)-\eta m g(2)-(1-\eta)m g(2)=0$.

\paragraph{2. Apply the doubling identity.}
Use the first circuit order in Fig.~\ref{fig:circuit}(a), with
$M=A_+B_+$. Lemma~\ref{lem:doubling}, at $p=\eta$, gives
\begin{equation}\label{eq:expandedzeroidentity}
 0=L_\eta(A_+,B_+)+L_\eta(A_-,B_-\mid M)
       +I(C_-;E_+\mid C_+)_\Xi.
\end{equation}
The lemma verifies the conditional-independence hypothesis needed for
the conditional deficit. All three terms are nonnegative, hence
\begin{equation}\label{eq:Jzero}
 I(C_-;E_+\mid C_+)_\Xi=0.
\end{equation}

\paragraph{3. Read the circuit as local thermal noise.}
By the commutation of the two kinds of beam splitters proved in
Lemma~\ref{lem:doubling}, the second circuit order has the same
four-output state $\Xi$. Immediately after the $\eta$-beam splitters,
the state is
\begin{equation}\label{eq:fourout}
 \omega_{CE}\otimes(\tau_2)_{G_C}\otimes(\tau_2)_{G_E}.
\end{equation}
Mix $E,G_E$ on the balanced beam splitter and retain $J:=E_+$.
The state before mixing $C,G_C$ is
\begin{equation}\label{eq:sigma}
 \sigma_{CJ}\otimes(\tau_2)_{G_C},\qquad
 \sigma_{CJ}=(\id_C\otimes\Phi)(\omega_{CE}).
\end{equation}
The auxiliary $G_C$ is independent of the entire pair $CJ$ because it
is a separate factor in~\eqref{eq:fourout}. After the remaining
balanced beam splitter, the $C_+J$ marginal is
\begin{equation}\label{eq:expandedlocalPhi}
 \Xi_{C_+J}=(\Phi\otimes\id_J)(\sigma_{CJ}).
\end{equation}

\paragraph{4. Convert zero conditional information into zero channel loss.}
Independence of $G_C$ and unitary invariance give
\begin{align}\label{eq:bigMI}
 I(C_+C_-;J)_\Xi
 &=I(CG_C;J)_{\sigma\otimes\tau_2}\notag\\
 &=S(\sigma_C)+S(\tau_2)+S(\sigma_J)-S(\sigma_{CJ})-S(\tau_2)\notag\\
 &=I(C;J)_\sigma.
\end{align}
On the other hand, the chain rule and~\eqref{eq:Jzero} give
\begin{align}\label{eq:smallMI}
 I(C_+C_-;J)_\Xi
 &=I(C_+;J)_\Xi+I(C_-;J\mid C_+)_\Xi=I(C_+;J)_\Xi.
\end{align}
Combining these identities with~\eqref{eq:expandedlocalPhi} yields
\begin{equation}\label{eq:expandedequalityDPI}
 I(C;J)_\sigma
 =I(C_+;J)_{(\Phi\otimes\id_J)(\sigma)}.
\end{equation}
Proposition~\ref{prop:strictMI} applies because $\sigma_C=\omega_C$
has finite energy and $S(\sigma_J)<\infty$. It makes $\sigma_{CJ}$ a product. Its marginals are
$\sigma_C=\omega_C$ and $\sigma_J=\Phi(\omega_E)$, so
\begin{equation}\label{eq:smoothedprod}
 (\id_C\otimes\Phi)(\omega_{CE})=\omega_C\otimes\Phi(\omega_E).
\end{equation}

\paragraph{5. Remove the noise on the other output.}
For $u,w\in\C^m$, write
$\widehat\omega(u,w)=\Tr\omega_{CE}[\D_C(u)\otimes\D_E(w)]$.
The channel formula~\eqref{eq:PhiCharacteristic} applies also to one
factor of a joint state. Equation~\eqref{eq:smoothedprod} therefore gives
\[
 \widehat\omega(u,w/\sqrt2)e^{-5|w|^2/4}
 =\widehat\omega_C(u)\widehat\omega_E(w/\sqrt2)e^{-5|w|^2/4}.
\]
The multiplier is nonzero, and $w\mapsto w/\sqrt2$ is onto $\C^m$.
Cancelling it yields
\[
 \widehat\omega(u,v)=\widehat\omega_C(u)\widehat\omega_E(v)
 \quad(u,v\in\C^m).
\]
Uniqueness of Weyl characteristic functions now proves
$\omega_{CE}=\omega_C\otimes\omega_E$.
\end{proof}

\subsection{Gaussian characterization and the converse}

Equation~\eqref{eq:outputdeficit} already gives (ii)$\Rightarrow$(i),
because its two deficits are nonnegative.

\begin{proof}[Proof of (ii)$\Rightarrow$(iii)]
The finite-energy quantum Darmois--Skitovich characterization,
Proposition~\ref{prop:quantumDS}, applies to the independent inputs and
product outputs. It gives Gaussianity and equality of covariance
matrices. Appendix~\ref{sec:DS} gives the full proof by reducing each
quadrature direction to two independent scalar random variables.
Their orthogonal mixtures are independent, which forces Gaussian laws
with equal variances. The means remain unrestricted.
\end{proof}

\begin{proof}[Proof of (iii)$\Rightarrow$(i), including product outputs]
Let the Gaussian inputs have covariance $V$ and means $d_A,d_B$.
Set $a=\sqrt\eta$, $b=\sqrt{1-\eta}$. The Gaussian output has
\begin{align}\label{eq:expandedGaussianblocks}
 d_C&=a d_A+b d_B,\qquad d_E=-b d_A+a d_B,\notag\\
 V_{CE}
 &=\begin{pmatrix}a^2V+b^2V&-abV+abV\\-abV+abV&b^2V+a^2V\end{pmatrix}\notag\\
 &=\begin{pmatrix}V&0\\0&V\end{pmatrix}.
\end{align}
Its characteristic function factorizes, so the outputs are a product.
Gaussian entropy depends only on the covariance; consequently
$S(C)=S(A)=S(B)$ and $L_\eta(A,B)=0$.
\end{proof}

\begin{remark}
Finite energy controls the circuit entropies and the scalar quadrature
variances. The restriction $0<\eta<1$ ensures that both mixing
coefficients are nonzero. At either endpoint, arbitrary inputs already
have zero linear deficit.
\end{remark}

An immediate consequence concerns optimization at fixed covariance.
\begin{corollary}[Uniqueness at fixed covariance]\label{cor:fixedcov}
Fix a physical covariance matrix $V$ and $0<\eta<1$. Among independent centered
$m$-mode input pairs whose two covariance matrices equal $V$, the linear deficit
$S(C)-\eta S(A)-(1-\eta)S(B)$ has the unique minimizer
$(\gamma_V,\gamma_V)$, with minimum zero. Here $\gamma_V$ is the
centered Gaussian state of covariance $V$.
\end{corollary}
\begin{proof}
The qEPI makes the deficit nonnegative. Theorem~\ref{thm:main} and
uniqueness of a Gaussian state with prescribed mean and covariance
identify its zero set.
\end{proof}

\section{Equality in the exponential qEPI}\label{sec:exp}

The exponential inequality itself is established: the balanced case is
\cite[Theorem~VII.2, Eq.~(89)]{KS}, and arbitrary transmissivity is
\cite[Eq.~(5)]{DPMG}. We derive its equality characterization from Theorem~\ref{thm:main}.

The proof has two distinct steps. We first propagate zero affine
deficit to a Gaussian limit. We then use strictness within the Gaussian
family to recover linear equality for the original inputs.

\subsection{The fixed affine weight}
Unlike classical differential entropy, bosonic entropy has no scaling
rule that converts a change of quadrature scale into an additive entropy
constant. One should therefore distinguish the single linear inequality
\eqref{eq:introlinear} from the full affine family equivalent to
\eqref{eq:introexponential}. We use that family explicitly when treating
exponential equality.

The equivalence between the complete affine family and exponential qEPI
is the variational reformulation used in
\cite[Eqs.~(106), (107), (115)]{DPT}. The following relative-entropy
identity spells out the elementary algebra and its equality condition.
Set
\begin{align}
 \mathcal Z_\eta&=\eta e^{S(A)/m}+(1-\eta)e^{S(B)/m},\notag\\
 \delta&=S(C)-m\log \mathcal Z_\eta,\qquad
 p_*=\frac{\eta e^{S(A)/m}}{\mathcal Z_\eta}.
\end{align}
The finite input entropies ensure $\mathcal Z_\eta>0$ and $p_*\in(0,1)$.
For a direct expansion, observe
\begin{align}
 \log p_*&=\log\eta+S(A)/m-\log \mathcal Z_\eta,\notag\\
 \log(1-p_*)&=\log(1-\eta)+S(B)/m-\log \mathcal Z_\eta.
\end{align}
Substitution into \eqref{eq:binaryD} gives
\begin{equation}
 m d(p\Vert p_*)=m d(p\Vert\eta)-pS(A)-(1-p)S(B)+m\log \mathcal Z_\eta.
\end{equation}
Adding $\delta=S(C)-m\log \mathcal Z_\eta$ proves
\begin{equation}\label{eq:affinedelta}
 F_p(A,B)=\delta+m\,d(p\Vert p_*).
\end{equation}
Thus exponential equality is $F_{p_*}=0$. Ordinary linear equality is
$F_\eta=0$. Given the established qEPI, \eqref{eq:affinedelta} implies
\begin{equation}
 L_\eta=0\quad\Longleftrightarrow\quad
 \delta=0\ \text{and}\ S(A)=S(B).
\end{equation}
The affine inequalities for all $p$ together are equivalent to the
exponential qEPI; the single case $p=\eta$ is the linear qEPI.

\subsection{Strict Gaussian entropy-power concavity}\label{sec:Gaussian}

Gaussian entropy is an average of log determinants over
$x\in(-1/2,1/2)$. This representation allows Minkowski's determinant
inequality to be applied pointwise. The entropy formula follows
from~\cite[Eq.~(3.27)]{HW}; the determinant inequality and its equality
case are recalled in Appendix~\ref{sec:detanalytic}. The following
Gaussian equality characterization can also be deduced from the EPnI
for Gaussian states~\cite[Appendix~A.8]{DePalma}, combined with the
strict concavity of $N\mapsto e^{g(N)}$ and of $V\mapsto S(\gamma_V)$.
We include a direct proof based on a log-determinant representation of
Gaussian entropy.

\begin{lemma}\label{lem:Gaussian}
Let $V_A,V_B$ be covariance matrices satisfying $V+i\Omega/2\geq0$,
and let $p\in[0,1]$ and $0<\eta<1$. Then
\begin{equation}\label{eq:Gzero}
 F_p(\gamma_{V_A},\gamma_{V_B})=0
 \quad\Longleftrightarrow\quad V_A=V_B\ \text{and}\ p=\eta.
\end{equation}
\end{lemma}

\begin{proof}
Recall $g$ from \eqref{eq:energyentropy}, including $g(0)=0$.
Williamson decomposition \eqref{eq:Williamson} gives
\begin{equation}\label{eq:integralS}
 S(\gamma_V)=m+\frac12\int_{-1/2}^{1/2}\log\det(V+ix\Omega)\,\dd x.
\end{equation}
Indeed, $V+ix\Omega=\mathsf S(\operatorname{diag}(\nu_1,\nu_1,\ldots)+ix\Omega)\mathsf S^T$.
Since $(\det\mathsf S)^2=1$, the determinant equals
$\prod_{j=1}^m(\nu_j^2-x^2)$. For each $\nu\geq1/2$, split
$\log(\nu^2-x^2)=\log(\nu-x)+\log(\nu+x)$ and integrate:
\begin{equation}
 1+\frac12\int_{-1/2}^{1/2}\log(\nu^2-x^2)\,\dd x
 =g(\nu-1/2).
\end{equation}
The left side equals
$1+\int_{-1/2}^{1/2}\log(\nu+x)\dd x
=(\nu+1/2)\log(\nu+1/2)-(\nu-1/2)\log(\nu-1/2)$.
This is $g(\nu-1/2)$. Endpoint integrals are improper when $\nu=1/2$,
but their logarithmic singularities are integrable. Summing and using
\eqref{eq:GaussianEntropy} proves \eqref{eq:integralS}.

For $|x|<1/2$, $V+ix\Omega$ is positive definite Hermitian. One way to see
this is to write it as a convex combination of $V+i\Omega/2$ and
$V-i\Omega/2$ with strictly positive coefficients; their kernels cannot
intersect because their sum is $2V>0$. Let
\begin{align}
 f_V(x)&=\det(V+ix\Omega)^{1/(2m)},\notag\\
 \mathcal G(f)&=\exp\left(\int_{-1/2}^{1/2}\log f(x)\,\dd x\right).
\end{align}
Minkowski's determinant inequality~\cite{HJ} is strict for positive
definite matrices unless they are positively proportional.
Apply this result pointwise to the positive Hermitian matrices
$V_A+ix\Omega$ and $V_B+ix\Omega$. It gives
\begin{equation}\label{eq:detstrict}
 f_{\eta V_A+(1-\eta)V_B}(x)
 \geq\eta f_{V_A}(x)+(1-\eta)f_{V_B}(x).
\end{equation}
Its equality condition for positive definite matrices is positive
proportionality. For $x\ne0$, proportionality
$V_A+ix\Omega=c(V_B+ix\Omega)$ forces $c=1$ by the imaginary parts,
and then $V_A=V_B$. Thus \eqref{eq:detstrict} is strict on a set of positive
measure whenever the covariance matrices differ.

We next use concavity of the integral geometric mean. If
$a=\mathcal G(f)>0$, $b=\mathcal G(h)>0$, and
$r=\eta a/(\eta a+(1-\eta)b)$, then
\begin{equation}
 \frac{\eta f+(1-\eta)h}{\eta a+(1-\eta)b}
 =r\frac{f}{a}+(1-r)\frac{h}{b}.
\end{equation}
Integrating pointwise log concavity proves
$\mathcal G(\eta f+(1-\eta)h)\geq\eta a+(1-\eta)b$.
All integrals are over the interval $(-1/2,1/2)$ of length one. The
logarithms of $f_{V_A}$ and $f_{V_B}$ are integrable by the factorization
above. The logarithm of their positive weighted sum is integrable as well:
it is bounded above, and bounded below by $\log\eta+\log f_{V_A}$.
Moreover $\mathcal G$ strictly increases under a pointwise increase on a set
of positive measure: the integral of the nonnegative log-ratio is then
strictly positive. Thus the strict pointwise determinant inequality survives
integration and exponentiation.
Since \eqref{eq:integralS} gives $e^{S(\gamma_V)/m}=e\mathcal G(f_V)$,
putting $V_C=\eta V_A+(1-\eta)V_B$ gives
\begin{equation}\label{eq:Qstrict}
 e^{S(\gamma_{V_C})/m}
 >\eta e^{S(\gamma_{V_A})/m}+(1-\eta)e^{S(\gamma_{V_B})/m}
\end{equation}
whenever $V_A\ne V_B$.
The affine identity \eqref{eq:affinedelta} excludes a zero affine deficit in
that case. When $V_A=V_B=V$, all three Gaussian entropies agree, so
$F_p(\gamma_V,\gamma_V)=m\,d(p\Vert\eta)$, proving \eqref{eq:Gzero}.
\end{proof}

\subsection{Quantum convolution and the finite-energy limit}
\label{sec:expandedCLT}

Only the following standard limit statements are needed: the
trace-norm quantum CLT~\cite[Theorem~5]{BDLR} and entropy continuity
under an oscillator energy bound~\cite[Lemma~15]{Winter}.
We verify their hypotheses for the present iterates below.

\paragraph{The $n$-fold quantum convolution.}
For an $m$-mode state $\rho$ and an integer $n\geq1$, take $n$ independent
copies of $\rho$.  Apply a passive unitary whose first output quadrature
vector is
\begin{equation}\label{eq:normalizedsum}
 R_{\mathrm{out},1}=\frac1{\sqrt n}\sum_{j=1}^n R_j,
\end{equation}
and discard the other outputs.  Here $R_j$ is the $2m$-component quadrature
vector on copy $j$.  Such a unitary exists: complete the unit vector
$n^{-1/2}(1,\ldots,1)$ to a real orthogonal matrix on the copy index and
apply this matrix to every quadrature component.  The resulting matrix is
both orthogonal and symplectic.  Denote the retained state by
$\rho^{\boxplus n}$, and put $\rho^{\boxplus1}=\rho$.
Independence of the inputs and commutation of observables on distinct
copies give
\begin{equation}\label{eq:convolutioncharacteristic}
 \chi_{\rho^{\boxplus n}}(\xi)
 =\chi_\rho(\xi/\sqrt n)^n,
 \qquad \xi\in\mathbb R^{2m}.
\end{equation}
This also shows that the retained state is independent of the chosen
orthogonal completion.  The state $\rho^{\boxplus n}$ is the $n$-fold symmetric quantum
convolution of $\rho$ studied in~\cite{BDLR,BGM,BM}; the characteristic
function~\eqref{eq:convolutioncharacteristic} shows that it coincides
with the $n$-fold convolution defined there. For $n=2$ it is
$\B_{1/2}(\rho\otimes\rho)$, and repeated balanced self-convolution gives
$\rho^{\boxplus 2^k}$ after $k$ steps.

If $\rho$ has displacement $d$ and covariance $V$, then
$\rho^{\boxplus n}$ has displacement $\sqrt n\,d$ and covariance $V$.
In particular, for a centered state, all these convolutions have the same
oscillator energy
\begin{equation}\label{eq:convolutionfixedenergy}
 \mathsf E=\Tr\rho H_m=\frac{\operatorname{tr}V-m}{2}.
\end{equation}
This is why centering precedes the limit argument.

Winter's Lemma~15~\cite{Winter} yields the explicit bound
\begin{equation}\label{eq:winterbound}
 |S(\rho)-S(\sigma)|
 \leq 2\epsilon m\,g\!\left(\frac{\mathsf E}{m\epsilon}\right)+h_2(\epsilon),
\end{equation}
where $\epsilon=\tfrac12\norm{\rho-\sigma}_1\in(0,1]$ and both oscillator
energies are at most $\mathsf E$. Here
$h_2(t)=-t\log t-(1-t)\log(1-t)$.
The right side tends to zero with $\epsilon$. The following lemmas record the precise continuity and limit statements.

\begin{lemma}[Entropy continuity under an oscillator energy bound]
\label{lem:entropycontinuity}
Let $\omega_j$ and $\omega$ be $m$-mode states such that
\[
 \norm{\omega_j-\omega}_1\longrightarrow0,\qquad
 \max\{\Tr\omega_jH_m,\Tr\omega H_m\}\leq\mathsf E<\infty.
\]
Then $S(\omega_j)\to S(\omega)$.
\end{lemma}

\begin{proof}
It follows directly from~\eqref{eq:winterbound}. 
\end{proof}

\begin{lemma}[Finite-energy convolution central limit theorem]
\label{lem:convolutionCLT}
Let $\rho$ be a centered $m$-mode state of finite energy and covariance
$V$.  Let $\gamma_V$ be the centered Gaussian state with covariance $V$.
Then
\begin{equation}\label{eq:CLTfullconclusion}
 \norm{\rho^{\boxplus n}-\gamma_V}_1\longrightarrow0,
 \qquad
 S(\rho^{\boxplus n})\longrightarrow S(\gamma_V).
\end{equation}
\end{lemma}

\begin{proof}
The finite-energy assumption gives all second moments. The
trace-norm convergence is the Cushen--Hudson theorem in the form
of~\cite[Theorem~5]{BDLR}; quantitative rates under higher moments are
given in~\cite{BDLR,BGM}.
Centering and~\eqref{eq:convolutionfixedenergy} keep every iterate,
and its Gaussian limit, in the same energy ball. Then Lemma~\ref{lem:entropycontinuity}
 gives the desired entropy convergence.
\end{proof}

\subsection{Proof of Theorem~\ref{thm:exp}}

\begin{proof}[Proof of Theorem~\ref{thm:exp}]
Different displacements of the inputs do not
alter any entropy, hence we can assume the displacements are all $0$. Fix $p=p_*$ from the original centered pair. Exponential
equality means $F_p(A,B)=0$. For $X=A,B,C$ and each integer $k\geq0$,
let $\rho_X^{(k)}=\rho_X^{\boxplus 2^k}$.
Apply Lemma~\ref{lem:doubling} to two independent copies of
$(\rho_A^{(k)},\rho_B^{(k)})$. Its plus pair is the next iterate, so
\begin{equation}\label{eq:iterationbound}
 0\leq F_p(A^{(k+1)},B^{(k+1)})
 \leq 2F_p(A^{(k)},B^{(k)}).
\end{equation}
The discarded terms are the conditional deficit and a CMI. The
conditional qEPI is used at this fixed $p=p_*$, which need not yet equal
$\eta$. Since the initial deficit is zero, induction gives
\begin{equation}\label{eq:iterzero}
 F_p(A^{(k)},B^{(k)})=0\qquad(k\geq0).
\end{equation}
That the output of the iterated pair is $\rho_C^{(k)}$ follows from
commutation of the two kinds of mixing.
The same $p$ is retained for all $k$.

For each $X=A,B,C$, the definition of the $n$-fold convolution and
independence give $d_{X^{(k)}}=0$ and $V_{X^{(k)}}=V_X$.
Here $V_C=\eta V_A+(1-\eta)V_B$ is fixed as well.
By \eqref{eq:energycov}, each of the three sequences has its own fixed finite
energy. Lemma~\ref{lem:convolutionCLT} therefore gives trace-norm convergence
to $\gamma_{V_A}$, $\gamma_{V_B}$, and $\gamma_{V_C}$, respectively, and
Lemma~\ref{lem:entropycontinuity} gives convergence of all three entropies.
The coefficient $p$ and the constant $m d(p\Vert\eta)$ remain unchanged.
We may consequently take the limit in the three finite entropy terms of
\eqref{eq:iterzero}, obtaining
\begin{equation}
 F_p(\gamma_{V_A},\gamma_{V_B})=0.
\end{equation}
Lemma~\ref{lem:Gaussian} yields $V_A=V_B$ and $p=p_*=\eta$.
The formula for $p_*$ then gives $S(A)=S(B)$, and
$L_\eta(A,B)=F_\eta(A,B)=0$. Theorem~\ref{thm:main} proves Gaussianity.

Conversely, matched-covariance Gaussian inputs have
$S(C)=S(A)=S(B)$ and satisfy exponential equality directly.
\end{proof}

\section{Discussion}\label{sec:discussion}

The equality condition for the bosonic qEPI can be expressed entirely
in terms of the two output systems: an independent finite-energy pair
attains either entropy inequality precisely when its beam-splitter
outputs are independent. The Gaussian characterization then follows
from the quantum Darmois--Skitovich theorem. The thermal auxiliary is
what connects an entropy condition on one output to this structural
condition on both outputs.

The method also separates two roles of Gaussian states. The auxiliary
thermal states are introduced at a fixed finite energy to obtain strict
information loss. The Gaussian limit in the exponential proof serves
a different purpose: it determines the affine weight and the input
covariances. The linear equality proof is complete before that limit
argument is used.

A natural extension of the present work is a quantitative version of
the qEPI: when the deficit is small, the input states should be close to
Gaussian states with nearly equal covariance matrices. The classical
counterpart is treated in~\cite{Courtade}, and stability of the
quantum Darmois--Skitovich theorem is treated in~\cite{Cuesta}. This quantitative version would have to control the information lost in the vacuum branch and the removal of thermal noise. In the present proof, the effect $G$ has eigenvalues tending to zero, and the inverse characteristic multiplier is unbounded. Exact injectivity therefore does not itself give a stability estimate. Energy-dependent estimates are a natural next question.
\section{Declarations of AI use}
Chat GPT 5.6 Sol was used to explore the design of  doubled circuit and also the proof strategy, especially the proof of Proposition 6. Chat GPT 6 Astra was used to draft the paper and refine the mathematical arguments. The proofs were independently checked and approved by the author, who takes full responsibility for the scientific content of the paper.
\appendix
\section{Auxiliary operator calculations}\label{sec:analytic}

\subsection{Moments and characteristic functions}
The products in~\eqref{eq:moments} are quadratic-form expectations.
Equivalently, $(V_\rho)_{jk}$ is the real part of the Hilbert--Schmidt
inner product between $(R_j-(d_\rho)_jI)\rho^{1/2}$ and
$(R_k-(d_\rho)_kI)\rho^{1/2}$. Their squared Hilbert--Schmidt norms are
finite under the second-moment assumption. This also justifies the
second moment of every scalar quadrature $\xi^TR$.

The Weyl representation is irreducible, so its linear span is
ultraweakly dense in the bounded operators. A trace-class operator
whose expectation against every Weyl operator vanishes is therefore
zero. This proves the uniqueness of characteristic functions used
in the main text.

\subsection{Oscillator entropy bound}
For $\beta>0$, nonnegativity of relative entropy to the Gibbs state gives
\[
 S(\rho)\leq \beta\Tr\rho H_m-m\log(1-e^{-\beta}).
\]
Minimizing over $\beta$ yields~\eqref{eq:energyentropy}.
For zero energy the state is the vacuum. Nonnegativity of the mode
number operators bounds each subsystem energy by the total energy.

\subsection{The heterodyne Fourier identity}\label{sec:FourierProof}

We verify~\eqref{eq:Qfourier}. Define the bounded operator
\begin{equation}
 T(w)=\int_{\C^m}e^{w\cdot\bar z-\bar w\cdot z}M_{\rm het}(\dd z)
\end{equation}
as a weak integral, meaning equality after taking matrix elements; its
integrand is a bounded scalar function against a POVM. The complex Gaussian
integral
$\int_{\C^m}\pi^{-m}e^{-|z|^2+A\cdot z+B\cdot\bar z}\dd^{2m}z
=e^{A\cdot B}$, for $A,B\in\C^m$, gives
\begin{equation}
 \langle u|T(w)|v\rangle
 =\langle u|v\rangle
 e^{\bar u\cdot w-\bar w\cdot v-|w|^2}.
\end{equation}
Normal ordering,
$\D(w)=e^{-|w|^2/2}e^{\sum_jw_ja_j^\dagger}e^{-\sum_j\bar w_ja_j}$,
shows that this is the matrix element of $e^{-|w|^2/2}\D(w)$.
Coherent vectors span a dense subspace, so the two bounded operators are
equal. Taking their expectation in $\rho$ proves \eqref{eq:Qfourier}.

The integral in~\eqref{eq:MP} is a Bochner integral, meaning a norm limit
of integrals of simple trace-class-valued functions. Its integrand is
trace-norm measurable and has norm $Q_\rho(z)$, an integrable probability
density. This establishes existence in trace norm.

\subsection{Integral maps and finite instrument entropies}
\label{sec:instrumentBounds}

For complete positivity of the maps in~\eqref{eq:channelBranches},
it suffices to check every finite-dimensional auxiliary system $R_0$.
If $T_{XR_0}\geq0$ is trace class, its heterodyne operator density is
\[
 A_{R_0}(z)=\pi^{-m}
   (\bra z\otimes I)T_{XR_0}(\ket z\otimes I)\geq0.
\]
Writing $B_z=\pi_z-c(z)\vac\geq0$, the amplified residual map is
\[
 (\Lambda_0\otimes\id_{R_0})(T)
   =\int_{\C^m}B_z\otimes A_{R_0}(z)\,\dd^{2m}z\geq0.
\]
The integral exists in trace norm: the integrand is measurable and
\[
 \int\Tr B_z\,\Tr A_{R_0}(z)\,\dd^{2m}z
 \leq\int\Tr A_{R_0}(z)\,\dd^{2m}z=\Tr T.
\]
It therefore defines a bounded completely positive map on trace-class
operators. The same argument applies to $\Lambda_1$ and their sum.
Only ordinary quantum channels on separable Hilbert spaces are used
in the data-processing steps~\cite[Theorem~1]{MR}.

We now justify the entropy chain rules in
Proposition~\ref{prop:strictMI}, using its notation. Both $M_f$ commute with $H_m$ and satisfy $M_f^2\leq I$,
so
\[
 p_f\Tr(H_m\omega^f_{X'})\leq\Tr(H_m\omega_X)<\infty.
\]
The Gibbs bound makes each $S(\omega^f_{X'})$ finite. The unflagged
output has energy
$\Tr(H_m\zeta_Y)=\tfrac12\Tr(H_m\omega_X)+m$.
Since $\zeta_Y=\sum_f p_f\zeta_Y^f$, each output branch has finite
energy and hence finite entropy as well. Finally,
$\omega_K=\sum_f p_f\omega_K^f$ and entropy concavity imply
$\sum_f p_f S(\omega_K^f)\leq S(\omega_K)<\infty$.
Subadditivity then gives finite entropy for every joint state in
Proposition~\ref{prop:strictMI}, including the states with the finite flag.

\subsection{Minkowski's determinant inequality}\label{sec:detanalytic}

For reference, we give the positive-definite Hermitian case of
Minkowski's determinant inequality~\cite{HJ}, including equality.
Let $H,K$ be positive definite Hermitian matrices of size $d=2m$ and
$0<t<1$. Let $\lambda_1,\ldots,\lambda_d>0$ be the eigenvalues of
$H^{-1/2}KH^{-1/2}$. Set $e_0(\lambda)=1$; for $1\leq k\leq d$ let
$e_k(\lambda)$ be the sum of all products of $k$ distinct eigenvalues.
Arithmetic--geometric mean over those products gives
\begin{equation}
 e_k(\lambda)\geq\binom dk\left(\prod_{j=1}^d\lambda_j\right)^{k/d}.
\end{equation}
Each eigenvalue occurs in $\binom{d-1}{k-1}$ products, which explains the
exponent $k/d$. Expanding the product and applying these inequalities yields
\begin{align}
 \prod_{j=1}^d[(1-t)+t\lambda_j]
 &=\sum_{k=0}^d(1-t)^{d-k}t^k e_k(\lambda)\notag\\
 &\geq\left[(1-t)+t\left(\prod_j\lambda_j\right)^{1/d}\right]^d.
\end{align}
Multiplication by $\det H$ proves
\begin{equation}
 \det((1-t)H+tK)^{1/d}
 \geq(1-t)\det(H)^{1/d}+t\det(K)^{1/d}.
\end{equation}
If equality holds, the $k=1$ arithmetic--geometric mean equality forces all
$\lambda_j$ to be equal (all expansion coefficients are positive).
Equivalently $K=cH$ for a positive scalar $c$; this also suffices for equality.

\section{Gaussian characterization from independent outputs}\label{sec:DS}

The characterization in Proposition~\ref{prop:quantumDS} is established:
it is the finite-energy instance of Cuesta's quantum Darmois--Skitovich
theorem~\cite[Theorem~7, Corollary~8]{Cuesta}. The one-mode predecessor is
\cite[Sec.~III]{Springer}. We retain the proof for completeness; the
reduction to one-dimensional quadrature laws follows the same principle
as Cuesta's proof. The scalar argument is stated first, followed by its application to
quadrature laws.

The scalar statement is the two-variable finite-variance
specialization of the classical Darmois--Skitovich
theorem~\cite{Darmois,Skitovich}; see~\cite{KLR} for the general
statement and its history. The finite-variance proof below is the
standard one based on the second derivatives of the logarithms of the
characteristic functions. Equality of the variances follows also
from the zero covariance of the orthogonal output forms.

\begin{lemma}[Scalar finite-variance Darmois--Skitovich lemma]\label{lem:scalarDS}
Let $X,Y$ be independent real random variables with finite second moments.
Let $a,b>0$ satisfy $a^2+b^2=1$. If $aX+bY$ and $-bX+aY$ are independent,
then $X,Y$ are Gaussian with the same variance, allowing distinct means.
Variance zero is permitted and means a constant random variable.
\end{lemma}
\begin{proof}
Let $\varphi(t)=\mathbb E e^{itX}$ and $\psi(t)=\mathbb E e^{itY}$.
The two independence assumptions give, for every $s,t\in\R$,
\begin{equation}\label{eq:DSfunctional}
 \varphi(as-bt)\psi(bs+at)
 =\varphi(as)\psi(bs)\varphi(-bt)\psi(at).
\end{equation}
The left side is the joint characteristic function of the output linear
forms; the right side is the product of their marginals.

Finite second moments justify two differentiations under the expectation:
$\varphi'(t)=i\mathbb E[Xe^{itX}]$ and
$\varphi''(t)=-\mathbb E[X^2e^{itX}]$, continuously in $t$, and similarly
for $\psi$. Since both characteristic functions equal one at zero, they
are nonzero on a small interval. There choose twice continuously
differentiable logarithms $\ell_X=\log\varphi$, $\ell_Y=\log\psi$ with $\ell_X(0)=\ell_Y(0)=0$.
On a sufficiently small connected neighborhood of $(0,0)$ the logarithm
of \eqref{eq:DSfunctional} is an exact additive identity: its two sides
can differ only by a continuous $2\pi i\mathbb Z$-valued function, which
vanishes at the origin and therefore throughout that neighborhood.

Differentiate that identity first in $s$ and then in $t$. Terms depending
only on $s$ or only on $t$ drop out, giving
\begin{equation}\label{eq:DShessian}
 -ab\ell_X''(as-bt)+ab\ell_Y''(bs+at)=0.
\end{equation}
The map $(s,t)\mapsto(u,v)=(as-bt,bs+at)$ is invertible and maps an open
neighborhood of the origin onto another such neighborhood. There is
therefore $r>0$ such that $\ell_X''(u)=\ell_Y''(v)$ whenever $|u|,|v|<r$.
Taking $v=0$, then $u=0$, shows that both second derivatives are constant.
At zero they equal minus the respective variances, so
$\operatorname{Var}X=\operatorname{Var}Y=:\sigma^2$. Integration gives
\begin{equation}\label{eq:DSlocal}
 \varphi(u)=e^{i\bar xu-\sigma^2u^2/2},\qquad
 \psi(u)=e^{i\bar yu-\sigma^2u^2/2}
 \quad(|u|<r),
\end{equation}
where $\bar x=\mathbb E X$ and $\bar y=\mathbb E Y$.

Local agreement of characteristic functions is not enough by itself.
To extend it, substitute $(s,t)=(au,-bu)$ and $(bu,au)$ into the original
global identity \eqref{eq:DSfunctional}. Using
$\varphi(-v)=\overline{\varphi(v)}$ and the analogous identity for $\psi$
yields the exact recursions
\begin{align}\label{eq:DSrecursions}
 \varphi(u)&=\varphi(a^2u)\varphi(b^2u)|\psi(abu)|^2,\notag\\
 \psi(u)&=\psi(a^2u)\psi(b^2u)|\varphi(abu)|^2.
\end{align}
Every argument on the right has absolute value at most $q|u|$, where
$q=\max\{a^2,b^2,ab\}<1$. If the Gaussian formulas hold on $[-R,R]$,
these recursions prove them on $[-R/q,R/q]$: the mean coefficient is
$a^2+b^2=1$ and the variance coefficient is
$a^4+b^4+2a^2b^2=1$. Start with any $0<R<r$ and iterate.
The intervals $[-R/q^n,R/q^n]$ exhaust the real line, so both formulas hold
globally. The uniqueness theorem for classical characteristic functions
identifies the two laws as the stated Gaussians.
\end{proof}

\begin{proposition}[Finite-energy quantum characterization]\label{prop:quantumDS}
If independent finite-energy $m$-mode inputs have independent outputs under
a beam splitter with $0<\eta<1$, they are Gaussian with equal covariance
matrices and arbitrary displacement vectors.
\end{proposition}
\begin{proof}
Fix any real vector $\xi\in\R^{2m}$. The self-adjoint quadrature
$\xi^TR_A$ has a classical spectral law in state $\rho_A$: its probability
of a Borel set $B\subseteq\R$ is
$\Tr\rho_A\mathbf1_B(\xi^TR_A)$, where $\mathbf1_B(\xi^TR_A)$ is its
spectral projection. Let $X$ have this law, and let $Y$ have the corresponding
law for $\xi^TR_B$. Their joint law in the product input is the independent
product of these two laws. Both second moments are finite; this follows
from the quadratic-form interpretation of \eqref{eq:moments} and
$\Tr\rho R_j^2<\infty$.

The two operators $\xi^TR_A$ and $\xi^TR_B$ act on different tensor factors,
so their spectral measures commute. Their two linear combinations therefore
have the ordinary joint law of
$aX+bY$ and $-bX+aY$, where $a=\sqrt\eta$, $b=\sqrt{1-\eta}$.
By \eqref{eq:intromixing}, this is also the joint spectral law of
$\xi^TR_C$ and $\xi^TR_E$ in the output state. Product outputs make these
laws independent. Lemma~\ref{lem:scalarDS} applies.

Thus, for each $\xi$, $X$ is Gaussian with mean $\xi^Td_A$ and variance
$\xi^TV_A\xi$, and $Y$ is Gaussian with mean $\xi^Td_B$ and the same
variance. Evaluating their scalar characteristic functions at one gives
\begin{equation}
 \chi_A(\xi)=e^{i\xi^Td_A-\xi^TV_A\xi/2},\qquad
 \chi_B(\xi)=e^{i\xi^Td_B-\xi^TV_B\xi/2}.
\end{equation}
These hold for every $\xi$, so the definition \eqref{eq:GaussianDef}
proves Gaussianity. The equality
$\xi^T(V_A-V_B)\xi=0$ for every real $\xi$ implies $V_A=V_B$ by
polarization of real quadratic forms. This is the finite-energy instance
of the exact quantum Darmois--Skitovich theorem \cite[Theorem~7]{Cuesta}.
\end{proof}

\end{document}